\documentclass[epsf,a4paper]{article}
\usepackage{fullpage}
\usepackage{amssymb,latexsym,amsmath}
\usepackage{graphicx}
\usepackage{comment}
\usepackage{amscd}
\usepackage{amsthm}
\usepackage{fontenc}

\usepackage[utf8x]{inputenc}
\newtheorem{theorem}{Theorem}%[section]
\newtheorem{coro}{Corollary}
\newtheorem{lem}{Lemma}

\newtheorem{definition}{Definition}
\newtheorem{remark}{Remark}
\newtheorem{counter}{Counter Example}
\newtheorem{proposition}[theorem]{Proposition}

\def\Xint#1{\mathchoice
{\XXint\displaystyle\textstyle{#1}}%
{\XXint\textstyle\scriptstyle{#1}}%
{\XXint\scriptstyle\scriptscriptstyle{#1}}%
{\XXint\scriptscriptstyle\scriptscriptstyle{#1}}%
\!\int}
\def\XXint#1#2#3{{\setbox0=\hbox{$#1{#2#3}{\int}$ }
\vcenter{\hbox{$#2#3$ }}\kern-.6\wd0}}

\def\dashint{\Xint-}

\def\({\left(}
\def\){\right)}
\def\1{\mathbf{1}}

\def\dt0{{{\frac{d}{dt}}_{|t=0}}}

\def\l|{\left|}

\def\r|{\right|}

\begin{document}
\title{Uniqueness results for critical points of a non-local isoperimetric problem via curve shortening}
\author{Dorian Goldman\thanks{Courant Institute of Mathematical Sciences/Universit\'{e} Paris 6 Pierre et Marie Curie, dgoldman@cims.nyu.edu}}
\maketitle

\begin{abstract}
Using area-preserving curve shortening flow, and a new inequality relating the potential generated by a set to its curvature, we study a non-local isoperimetric problem which arises in the study of di-block copolymer melts, also
referred to as the Ohta-Kawasaki energy. We are able to show that the only connected critical point is the ball under mild assumptions on the boundary, in the small energy/mass regime. In particular this class includes all rectifiable, connected 1-manifolds in $\mathbb{R}^2$. We also classify the simply connected critical points on the torus in this regime,
showing the only possibilities are the stripe pattern and the ball.  In $\mathbb{R}^2$, this can be seen as a partial union of the well known
result of Fraenkel \cite{Fraenkel} for uniqueness of critical points to the Newtonian Potential energy, and Alexandrov for the perimeter functional \cite{alexandrov}, however restricted to the plane. 
 The proof of the result in $\mathbb{R}^2$ is analogous to the curve shortening result due to Gage \cite{Gage2}, but
involving a non-local perimeter functional, as we show the energy of convex sets strictly decreases along the flow.
Using the same techniques we obtain a stability result for minimizers in $\mathbb{R}^2$ and for the stripe pattern on the torus, the latter of which was recently shown to be the global minimizer to the energy
when the non-locality is sufficiently small \cite{sternberg}.
\end{abstract}

\section{Introduction} 

The classical isoperimetric problem has been thoroughly studied, and it is well known since the work of De Giorgi \cite{degiorgi} that
the unique optimizer to this problem in $\mathbb{R}^d$ is the ball. There has recently been significant interest in the effects of
adding a repulsive term to the classic perimeter which favors separation of mass. An example of such an energy is often referred
to as the Ohta-Kawasaki energy, first introduced in \cite{ohta86}, and takes the following form 
\begin{equation}\label{energy3}
 E[u] := \int_{U} |\nabla u|+ \gamma \int_{U}\int_{U} (u(x)-\bar u)G(x,y) (u(y)-\bar u) dxdy,
\end{equation}
where $u \in BV(U ;\{0,+1\})$ and 
\begin{equation}\label{ubardef}
 \bar u = \left\{\begin{array}{ccc}
\dashint_{U} u(x) dx &\mbox{for}&U \textrm{ bounded and open }\\
0&\mbox{for}&U=\mathbb{R}^d.\\
\end{array}\right.
\end{equation}

Here $BV(U;\{0,+1\})$ denotes the space of functions of bounded variation taking values $0$ and $+1$ in $U$ (see \cite{ambrosio} for an introduction to the space $BV$). The kernel $G$ 
is generally taken to be the kernel of the Laplacian operator, with periodic boundary conditions when $U=\mathbb{T}^d$.  The parameter $\gamma > 0$ describes the strength of the non-local term. It is clear that the non-local
term favors the separation of mass while the perimeter favors clustering. The above problem describes a number of polymer systems \cite{degennes79, nagaev95,ren00} as well
as many other physical systems \cite{chen93,emery93,glotzer95,lundqvist,kohn07iciam,nagaev95} due to the fundamental nature of the Coulombic term. Despite
the abundance of physical systems for which \eqref{energy3} is applicable, rigorous mathematical analysis is fairly recent \cite{alberti09,choksi01,choksi08,choksi04,choksi11siads,choksi10,choksi11,gms11b,gms10b,m:phd,m:pre02,m:cmp10,muratovknup1,muratovknup2}.

\noindent One considers minimizing \eqref{energy3} over the class  \begin{align}\label{Energy1}
\mathcal{A}_m &:= \{ u \in BV(U;\{0,+1\}) : \int_U u \;dx = m\}, \;\;\;\;\;u := \chi_{\Omega}.
\end{align}
The variational problem \eqref{energy3} when $U=\mathbb{T}^2$ has been studied in the limit $\gamma \to +\infty$ when the phase $\{u=0\}$ dominates the $\{u=+1\}$ \cite{choksi10,choksi11,gms10b,gms11b,m:cmp10}. The
minimizers in this case form droplets of the minority phase $\{u=+1\}$ in the majority phase $\{u=0\}$ and each connected component of $\{u=+1\}$ wishes to minimize the energy
\begin{align}\label{energy4}
E_{\mathbb{R}^d} [u]= \int_{\mathbb{R}^d} |\nabla u|+ \int_{\mathbb{R}^d}\int_{\mathbb{R}^d}u(x)G(x,y)u(y) dxdy,
\end{align}
over $u \in \mathcal{A}_m$, where $G$ is a kernel on $\mathbb{R}^d$ to be explicited below and $m$ is $O(1)$ as $\gamma \to +\infty$. When we refer to \eqref{energy3} with $U=\mathbb{R}^d$ we will
mean \eqref{energy4} throughout the paper. Observe that we no longer have the constant $\gamma$ in front of the non-local term in this case, as one can rescale the the domain $\Omega$ to make $\gamma \equiv 1$ for the class of $G$ we consider (see \eqref{kerndef1} below), changing
the energy by at most a constant. Indeed when $\gamma \neq 1$, by a change of
variables $x' = \lambda x$, $y'=\lambda y$ and setting $\tilde u(x) = u(x/\lambda)$ we can write \eqref{energy3} as
\begin{align}\label{energy4aa}
E [u]= \lambda \(\int_{\mathbb{R}^d} |\nabla \tilde u|+ \gamma \lambda^{3-p} \int_{\mathbb{R}^d}\int_{\mathbb{R}^d}\tilde u(x)G(x,y)\tilde u(y) dxdy + C(\lambda)\),
\end{align}
 where $C(\lambda)>0$ is a constant which arises when $G$ is the logarithmic kernel, and $p > 0$ arises from the singularity of the kernel. Setting $\lambda^{3-p}\gamma=1$ yields \eqref{energy4} after dividing by $\lambda$ and subtracting the constant term. We will often abuse terminology and say $\Omega \in \mathcal{A}_m$ when we mean $\chi_{\Omega} \in \mathcal{A}_m$ and write $E(\Omega)$ when we mean $E(\chi_{\Omega})$. For
the remainder of the paper we will study \eqref{energy3} exclusively  with $U=\mathbb{T}^2$ and $U=\mathbb{R}^2$, except for comments regarding
results in $\mathbb{R}^d$ for $d \geq 3$. 

\medskip

 When $U=\mathbb{R}^2$, in order for minimizers to exist in $\mathcal{A}_m$ to \eqref{energy3}, it is necessary to modify the logarithmic kernel. The above energy
with $G$ replaced by the kernel 
\begin{equation}\label{Kkern}
 K(x) = \frac{1}{|x|^{\alpha}} \textrm{ for } \alpha \in (0,2),
\end{equation}
was recently studied by Knupfer and Muratov \cite{muratovknup1,muratovknup2}. 
A simple rescaling shows that for small masses, the effect
of the non-local term is small compared to that of the perimeter. It is therefore reasonable to expect that for small masses, the unique
minimizer to \eqref{energy3} is the ball. This was shown by Knupfer and Muratov \cite{muratovknup1, muratovknup2} in dimensions
$2 \leq d \leq 8$. Moreover they show that for sufficiently large masses, for $2 \leq d \leq 8$, minimizers of \eqref{energy3} fail to exist, as it is favorable for mass to split.\\

\medskip

There is also much interest in critical points to \eqref{energy3} which are not necessarily locally minimizing. Here a critical point to \eqref{energy3}
is a set $\Omega \in \mathcal{A}_m$ for which the first variation with respect to volume preserving diffeomorphisms vanish. 
\begin{definition}\label{defcp}
A set $\Omega \in \mathcal{A}_m$ is a \emph{critical point} of \eqref{energy3} if for all volume preserving diffeomorphisms $\phi_t:\Omega \to \phi_t(\Omega) =: \Omega_t$ it holds that
\begin{equation}
\frac{d E(\Omega_t)}{dt}\big|_{t=0} = 0.\end{equation}
\end{definition}

 In particular a simple calculation \cite{choksi12} shows that $\Omega$ is a smooth critical point if and only if it solves the following Euler-Lagrange equation
\begin{equation}\label{ELeqn}
 \kappa(y) + \gamma \phi_{\Omega} (y) = \lambda \textrm{ on } \partial \Omega
\end{equation}
where $\phi_{\Omega}$ is the potential generated by $\Omega$, ie. 
\begin{equation}\label{potdefn1}
 \phi_{\Omega}(y) = \left\{\begin{array}{ccc}
\int_{\Omega} G(x,y) dx &\mbox{for}&U=\mathbb{R}^2\\
\int_{\mathbb{T}^2} G(x,y)(u(x)-\bar u) dx &\mbox{for}&U=\mathbb{T}^2,\\
\end{array}\right.
\end{equation}
$\kappa$ is the curvature of $\partial \Omega$ and $\lambda$ is the Lagrange multiplier arising from the volume constraint. We however
do not wish to assume a-priori regularity of the boundary as critical points may in general not be smooth. An important example demonstrating this
is the coordinate axes in $\mathbb{R}^2$ minus the origin. In this case the generalized mean curvature
is constant on the reduced boundary $\partial^*\Omega = \Omega$, $m=0$ and hence \eqref{ELeqn} is satisfied everywhere on $\partial^* \Omega$ (see \cite{simon} for a reference
on generalized mean curvature and the theory of varifolds). The example of a figure 8 with center at $O$ also shows that while the curvature of a closed curve can be smooth and uniformly bounded on $\partial \Omega \backslash \{O\} = \partial^* \Omega$,
there is not necessarily a natural way to make sense of the curvature or variations of \eqref{energy3} near $O$. In $\mathbb{R}^2$ however we can continue to make sense of the curvature
at $O$ if there exists a parametrization of the boundary by a closed, rectifiable (ie. has finite length) curve. By the results of \cite{ambrosio2}, any connected
set with finite perimeter has a boundary $\partial \Omega$ which can be decomposed into a countable union of Jordan curves $\{\gamma_k\}_k$ with disjoint interiors. In this
case however, as the figure 8 example demonstrates, one cannot make sense of variations near points on the curve which are locally homeomorphic to $[0,1]$, and can thus only expect to extract information
from the Euler-Lagrange information on the reduced boundary. However, when the boundary $\partial \Omega$ can be
decomposed into a countable \emph{disjoint} union of closed, rectifiable curves, there is a natural way to consider variations
of the domain, even on the compliment of the reduced boundary, by considering variations of the curve in the normal direction induced by the parametrization. We thus make the following definition.

\begin{definition} (Admissible curves) \label{addef} A connected set $\Omega$ with finite perimeter will be called \emph{admissible} if its boundary $\partial \Omega$ can be decomposed
into a countable number of closed, disjoint curves $\gamma_k$ each admitting a $W^{1,1}$ parametrization $\gamma_k:[0,1] \to \mathbb{R}^2$ with
$|\gamma_k'(t)|=L_k$ for $t \in [0,1]$. In particular we may write $\gamma_k(t) = \gamma_k(0) + L_k\int_0^t e^{i\theta_k(r)} dr$. \end{definition}

Note that the above class includes all connected, rectifiable 1-manifolds. Indeed any manifold has a boundary which is not locally homeomorphic to $[0,1]$
and thus does not intersect any other segments of the boundary. In particular, each $\gamma_k$ is therefore simple and disjoint from every other $\gamma_j$ for $j \neq k$, and
never intersects itself transversally. Working within the class of admissible curves, we are able to rigorously compute the Euler-Lagrange equation and extract sufficient information from it to conclude
that admissible critical points are convex when $U=\mathbb{R}^2$ and simply connected critical points are convex or star shaped when $U=\mathbb{T}^2$, in the small mass/energy regime. The details
will be presented in Section \ref{weakeleqn}.

\medskip

\noindent When $U=\mathbb{R}^2$, the classification
of critical points corresponding to either the perimeter term or non-local term, considered separately, has been well studied \cite{alexandrov, Fraenkel}. In particular
it is a well known result of Alexandrov \cite{alexandrov} that in dimensions $d \geq 2$ the only simply connected, compact, constant mean curvature surface is the ball. For
the non-local term, Fraenkel showed somewhat recently \cite{Fraenkel} that if $\phi_{\Omega} \equiv $ constant on $\partial \Omega$ then $\Omega$ must be the ball. This was recently extended to general Riesz kernels by Reichel \cite{Reichel}, however restricted to the class of convex sets.
The question then naturally arises of knowing how one can classify the solutions to \eqref{ELeqn}. In $\mathbb{R}^2$ one can easily construct annuli which satisfy \eqref{ELeqn} for particular choices
of radii (see Counter Example \ref{CE1}). Even in $\mathbb{R}^3$  examples of tori and double tori solutions to \eqref{ELeqn} exist  \cite{toroidal}, showing that
compact, connected solutions to \eqref{ELeqn} exist other than the ball. Since a smooth set is a critical point in the sense of Definition \ref{defcp} if and only if it satisfies \ref{ELeqn} (see \cite{choksi12}), this equivalently
shows a lack of uniqueness for critical points of \eqref{energy3}.  We provide a partial answer to this question (see Theorem \ref{main3} below) for a range of values of $(m,E) \in \mathbb{R}^+ \times \mathbb{R}^+$ sufficiently small, by showing that the only connected critical point is the ball in
the class of admissible sets in that range. More precisely, when the parameters \eqref{rescaled1} or \eqref{rescaled2} are sufficiently small. \\

It is easily seen that the only constant curvature surfaces in $\mathbb{T}^2$ are unions of circular arcs and straight lines. In addition, the stripe patterns defined by
\begin{equation}
 u_n(x) = u(nx) \textrm{ for } n \in \mathbb{N},\label{R.5}
\end{equation}
where
\begin{equation}\label{R.6}
u(x)=\left\{\begin{array}{ccc}
1&\mbox{for}&0\le x\le w\\
0&\mbox{for}&w\le x\le 1\\
\end{array}\right\},
\end{equation}
for $w \in (0,1)$ also satisfy $\phi_{\Omega} = $ constant on $\partial S_n$ where $S_n$ is the set corresponding to the indicator function $u_n$. We are also able to classify simply connected solutions to \eqref{ELeqn} in this case, showing that in fact, for a range of $(\gamma,E) \in \mathbb{R}^+ \times \mathbb{R}^+$ sufficiently small, there are no solutions to \eqref{ELeqn} other than $S_{n=1}$ and the ball.
In particular, when $m > \frac{\pi}{4}$, the only possibility is $S_{n=1}$. 
\medskip

 The way that we characterize critical points in both cases is by showing that any set $\Omega$ as described above which
is not a constant curvature surface satisfies
\begin{equation}\label{notcr} \frac{dE(\Omega_t)}{dt}\big|_{t=0} \neq 0,\end{equation}
where $\Omega_t$ is the evolution of $\Omega$ under area-preserving curve shortening flow, which is admissible under Definition \ref{defcp}. Details will follow in Sections \ref{VPMCFsec} and \ref{mainresults}.\\

 \noindent When $U=\mathbb{R}^2$ our main results (Theorems \ref{main3} and \ref{nonexistI} below) will hold
for \begin{equation}\label{kerndef1} G(x,y) = -\frac{1}{2\pi} \log |x-y| \textrm{    or    } G(x,y) = \frac{1}{|x-y|^{\alpha}} \textrm{ when } \alpha \in (0,1),\end{equation} with minor modifications to the proofs. It will
always be made clear below which kernel is being used. When a constant depends on $\alpha$, this will mean exclusively for the kernel $K$ for $\alpha \in (0,1)$. In all such cases
the dependence of the constant on $\alpha$ may be dropped for the logarithmic kernel. \\

 \noindent We begin by defining the following parameters 
\begin{align}\label{rescaled1}
\bar \eta  &:=\left\{\begin{array}{ccc}
 m^{1/2} L^2 (1+|\log L|)&\mbox{for}&G(x,y) = - \frac{1}{2\pi} \log |x-y|,\;\;\; \gamma \equiv 1 \\
 m^{1/2} L^{2-\alpha}&\mbox{ for }&G(x,y) = \frac{1}{|x-y|^{\alpha}} \;\;\; \alpha \in (0,1),\;\;\;  \gamma \equiv 1\\
\end{array}\right.\\
\bar \gamma &:= \gamma m^{1/2} L^2 (1+|\log L|)\label{rescaled2},\;\;\;  m \in (0,1)
\end{align}
where $L=|\partial \Omega|$. Our results will be stated in terms of these rescaled parameters. A simple scaling analysis of \eqref{energy3} reveals  
\begin{equation}
 \left\langle \frac{d (E-L)}{dL}, \zeta \right\rangle = \gamma \frac{ \int_{\partial \Omega} \phi_{\Omega} (y) \zeta(y) dS(y)}{\int_{\partial \Omega} \kappa(y)\zeta(y)dS(y)} \sim \bar \eta , \bar \gamma,
\end{equation}
where $E$ is defined by \eqref{energy3}, $dS$ is surface measure on $\partial \Omega$ and with some abuse of notation  $\left\langle \frac{d (E-L)}{dL}, \zeta \right\rangle$ denotes
the variation in the sense of Definition \ref{defcp} induced by the normal velocity $\zeta:\partial \Omega \to \mathbb{R}$. Thus $\bar \eta $ and $\bar \gamma$ represent the 
rate of change of the non-local term in the energy with respect to the length of the boundary. Our result can thus be stated
formally as saying that when the the change of the non-local term is small compared to a change in length, the critical points
can be classified entirely in terms of those of the length term in \eqref{energy3}, and thus are constant curvature curves. 

 For minimizers we have a natural a priori bound on $L$ coming from the positivity of both terms
in the energy \eqref{energy3}, which we don't have for non-minimizing critical points. This explains the need for introducing \eqref{rescaled1}--\eqref{rescaled2}. The terms $\bar \eta _{cr}$ and $\bar \gamma_{cr}$ below are critical values of $\bar \eta $ and $\bar \gamma$ which can be made explicit and are described in more detail in Section \ref{mainresults}.
\begin{theorem}\label{main3} When $U=\mathbb{R}^2$  there exists $\bar \eta _{cr} = \bar \eta _{cr}(\alpha) > 0$ such that whenever $\bar \eta  < \bar \eta _{cr}$, the only admissible critical point of \eqref{energy3} in $\mathcal{A}_m$ in the sense of Definition \ref{addef} is the ball. When $U=\mathbb{T}^2$ there exists a $\bar \gamma_{cr}>0$ such that whenever $\bar \gamma < \bar \gamma_{cr}$, the only simply connected critical
points to \eqref{energy3} in $\mathcal{A}_m$ for all $m \in (0,1]$ are the ball and the stripe pattern $S_{n=1}$ defined by \eqref{R.5}-\eqref{R.6}. In particular when $m > \frac{\pi}{4}$, the only simply connected critical point is $S_{n=1}$. \end{theorem}

The first part of the Theorem when $U=\mathbb{R}^2$ is false when the mass is larger, as annuli
satisfying \eqref{ELeqn} can easily be constructed. %The above
%theorem however leaves open the possibility that there exist other non-connected, compact
%solutions to \eqref{ELeqn}. \\
\begin{counter}\label{CE1}
 There exists a smooth, compact, connected set $\Omega$ solving \eqref{ELeqn} with $\bar \eta  > \bar \eta _{cr}$ which is not the ball. 
\end{counter}

\begin{remark} It is easy to see that there cannot exist critical points with multiple disjoint components which are separable by a hyperplane. Indeed if $\Omega_1$ and $\Omega_2$ are two
disjoint components of $\Omega$ which can be separated, then let $c$ be a vector so that $(x-y) \cdot c > 0$ for $(x,y) \in \Omega_1 \times \Omega_2$. Then we have
\begin{equation}
 \frac{d}{dt} \iint_{\Omega_1 \times \Omega_2 } \log |x+ ct - y| dxdy  = - \iint_{\Omega_1 \times \Omega_2} \frac{(x-y)}{|x-y|^2} \cdot c dxdy < 0.
\end{equation}
 Consequently $\Omega$ cannot be critical in the sense of Definition \ref{defcp}. 
Theorem \ref{main3} however leaves open the possibility of more intricate solutions to \eqref{ELeqn}, with multiple connected components. A similar calculation shows the same result for the kernel $K(x)=1/|x|^{\alpha}$.
\end{remark}

Using the same techniques we also obtain the following stability results. The first concerns global minimizers on the torus $\mathbb{T}^2$ and can be seen as a statement about the stability of the recent result of Sternberg and Topaloglu \cite{sternberg}. We consider
the class of connected sets belonging to $\mathcal{A}_m$
\begin{equation}
 \mathcal{A}_m^c = \{ u \in \mathcal{A}_m : \Omega \textrm{ is simply connected }\}.
\end{equation}

\noindent We then have the following theorem. Note that we use the original parameter $\gamma$.
\begin{theorem} \label{main5} When $m = \frac{1}{2}$ and $U=\mathbb{T}^2$, there exists a $\gamma_{cr}>0$, a functional $F:\mathcal{A}_{\frac{1}{2}}^c \to \mathbb{R}$ and $C=C(\gamma_{cr})>0$ such that whenever $\gamma < \gamma_{cr}$ 
\begin{equation}
 E(\Omega) \geq E(S_{n=1}) + CF(\Omega),
\end{equation}
 where $F(\Omega) \geq 0$ with equality if and only if $\Omega=S_{n=1} := [0,\frac{1}{2}] \times [0,1]$. 
\end{theorem}

\noindent The above theorems all rely on Theorem \ref{main4} in Section \ref{mainresults} which provides an explicit estimate
for the rate of decrease of the energy \eqref{energy3} along area-preserving curve shortening flow. 
The following stability result for minimizers in $\mathbb{R}^2$ is a simple corollary of Theorem \ref{main4}.
\begin{coro}\label{R2stable}
 Let $\Omega$ be any convex set in $\mathbb{R}^2$. Then there exists an $m_{cr} >0$ such that whenever $m < m_{cr}$ there is a constant $C=C(m_{cr})>0$ such that
\[ E(\Omega) \geq E(B) + C (L - 2\sqrt{\pi} m^{1/2}),\]
where $L=|\partial \Omega|$.
\end{coro}

\medskip

\noindent The main geometric
inequality which we prove in this paper in order to control the non-local term along the flow is the following.

 %A variant of \eqref{strong2} cannot
%hold in higher dimensions since sets can have very long spikes with small surface area. However restricted to convex sets, a version of \eqref{strong2} can be shown
%to hold \cite{groemers}.\\
%\medskip
\begin{theorem}\label{nonexistI}

Let $\Omega \subset U$ be a convex set with $\kappa \in L^2(\partial \Omega)$ and $U=\mathbb{R}^2$ or $U=\mathbb{T}^2$. Then there is an explicit constant $C=C(|\partial \Omega|,\alpha)>0$ such that
\[\|\phi_{\Omega} - \bar \phi_{\Omega} \|_{L^{\infty}(\partial \Omega)}^2 \leq C \int_{\partial \Omega}|\kappa - \bar \kappa|^2dS,\]
where the bar denotes the average over $\partial \Omega$ and $dS$ is 1-dimensional Hausdorff measure.  The above continues to hold in $\mathbb{T}^2$ for any set homeomorphic to $S_{n=1}$. 
                                                                                         
\end{theorem}
The above theorem provides a quantitative estimate of the closeness to an equipotential surface in terms of the distance to a constant curvature surface. 
The inequality in $\mathbb{R}^2$ in fact relies on an isoperimetric inequality due to Gage \cite{Gage2} applied to curve shortening for convex sets.
We hope that the above inequality will
be of interest even outside the context of Ohta-Kawasaki. 
\begin{remark}
We remark that if $\Omega$ is any connected set with $\kappa \in L^2(\partial \Omega)$ we can prove the weaker inequality
\begin{equation}\label{weakguy}
 \|\phi_{\Omega} - \bar \phi_{\Omega} \|_{L^{\infty}(\partial \Omega)}^2 \leq C\sqrt{ \int_{\partial \Omega}|\kappa - \bar \kappa|^2dS}.
\end{equation}
Indeed if one follows the proof of Theorem \ref{nonexistI}, one can apply Cauchy-Schwarz on line \eqref{SCeqn} and bound 
$\int_{\partial \Omega} p^2 dS$ by $CL^2$, thus establishing \eqref{weakguy} with $C \sim L^3$. This inequality turns out not to be sufficient to show that the energy \eqref{energy3} decreases along the
flow however. Observe that neither inequality can hold without the assumption of connectedness, as the example of two disjoint balls 
demonstrates.
 \end{remark}
\begin{remark}
 In dimensions $d \geq 3$, we expect the above inequality to continue to hold, but are unable to demonstrate it without an assumption that
the sets $\Omega$ satisfy a positive uniform lower bound on the principal curvatures of the surface $\partial \Omega$. In this case the constant $C$ also depends on this
lower bound. Proving Theorem \ref{nonexistI} is the only obstacle in extending our results to the case $U=\mathbb{R}^3$. The proof presented in Section \ref{Rn} fails
in $\mathbb{R}^d$ for $d \geq 3$ since the Gaussian curvature and mean curvature do not agree.
\end{remark}

Our paper is organized as follows. In Section \ref{weakeleqn} we set up the appropriate framework for critical points, defining precisely
in what sense \eqref{ELeqn} is satisfied and showing that when $U=\mathbb{R}^2$ critical points are convex when $\bar \eta $ is sufficiently small. In Section \ref{VPMCFsec} we introduce the area-preserving curve shortening flow and state some of the main
results concerning the flow that we will need. In Section \ref{mainresults} we state precisely the result showing \eqref{notcr}, Theorem \ref{main4}. In Section \ref{stability}
we establish the necessary inequalities needed to control the behavior of the non-local term in terms of the decay of perimeter along the flow (cf. Theorem \ref{nonexistI}). Finally we use the geometric inequalities established in Section \ref{stability} to prove the above theorems in Section \ref{finalproof}
by differentiating the energy \eqref{energy3} along the flow. The counter example (cf. Counter Example \ref{CE1}) will appear at the end of Section \ref{stability}.

\section{The weak Euler-Lagrange equation}\label{weakeleqn}
In this section we rigorously compute the Euler-Lagrange equation for the class of curves admissible in the sense of Definition \ref{addef}. We work in $\mathbb{R}^2$ for simplicity of presentation and hence set $\gamma \equiv 1$. The calculation of
the Euler-Lagrange equation is essentially identical on $\mathbb{T}^2$ however the analysis of critical points differs slightly and so we reserve this for Section \ref{Torus}. 
\medskip

\noindent In the class of admissible curves (cf. Definition \ref{addef}) the energy \eqref{energy3} may be written as
\begin{equation}\label{energy4a}
 E(u) = \sum_k \int_0^{L_k} |\gamma_k'(s)|ds + \iint_{\Omega \times \Omega} G(x-y) dx dy,
\end{equation}
where $\gamma_k$ is as in Definition \ref{addef}, and $|\gamma_k'(s)|=1$ for a.e $s \in [0,1]$. Since $\gamma_k \cap \gamma_j = \emptyset$ when $j \neq k$, the variations $\gamma_k \mapsto \gamma_k + t v$ such that
$\int_0^{L_k} v(s) \cdot (\gamma_k'(s))^{\perp}ds = 0$ are admissible for $t>0$ sufficiently small in Definition \ref{defcp} by letting $\Omega_t = \textrm{ Int }(\gamma_k + tv) $, where Int$(\gamma)$ denotes
the interior of the closed curve $\gamma$.
\begin{proposition} \label{weakcomp} (Weak Euler-Lagrange equation) Let $\Omega$ be a critical point to \eqref{energy4a} in the sense of Definition \ref{defcp}, which is admissible in the sense of Definition \ref{addef}.
Then for every $k$ it holds that
 \begin{align}\label{realel}
 \kappa(\gamma_k(s)) + \phi_{\Omega} (\gamma_k(s)) = \lambda
\end{align}
where $\kappa(\gamma_k(s)) = \gamma_k''(s) \cdot (\gamma_k'(s))^{\perp}$ is the curvature and $\lambda$ is the Lagrange multiplier from
the volume constraint. Moreover $\gamma_k \in C^{3,\alpha}([0,1])$ for all $k$. 
\end{proposition}
\begin{proof}
By taking variations $t \mapsto \gamma_k + tv$ as described above and differentiating \eqref{energy4a} with respect to $t$, we have
\begin{equation}\label{weakel}
 \int_0^{1} \gamma_k'(s) \cdot v'(s) ds + L_k\int_0^{1} \phi_{\Omega}(\gamma_k(s))  v(s) \cdot (\gamma_k'(s))^{\perp} ds = 0,
\end{equation}
for all $v \in W^{1,\infty}([0,1];\mathbb{R}^2)$ where we've re-parametrized so that $|\gamma_k'(t)|=L_k$ and $\gamma_k:[0,1] \to \mathbb{R}^2$. 
Equation \eqref{weakel} is the weak Euler-Lagrange equation for \eqref{energy4a}.  Observe that then
\[ v \mapsto \int_0^{1} \gamma_k'(s) \cdot v'(s) ds,\]
is a bounded linear functional on $W^{1,\infty}([0,1])$ which extends continuously to a bounded linear functional on $C^0([0,1])$. Indeed this follows from \eqref{weakel}, since $\phi_{\Omega}(\gamma_k(s)) \in C^{1,\alpha}([0,1])$ \cite{gilbarg}
and $\gamma_k \in W^{1,1}([0,1])$.
Thus by the Riesz representation theorem, $\gamma_k''$ is a finite, vector valued Radon measure on $[0,1]$. In fact,
since $\gamma_k''(s) = (\lambda - L_k\phi_{\Omega}(\gamma_k(s)))(\gamma_k'(s))^{\perp}$ as a measure, it follows that $\gamma_k' \in W^{1,1}([0,1])$.  
Recalling that $\gamma_k(s) = \gamma_k(0) + L_k\int_0^s e^{i\theta(r)}dr$, we have $\gamma_k''(s) = L_k 
\theta'(s) e^{i\theta(s)}$ for a.e $s \in [0,1]$ since $\gamma_k' \in W^{1,1}([0,1])$ and hence $|\gamma_k''(s)| = L_k|\theta'(s)|$ a.e. This implies $\theta' \in L^1([0,1])$.
Then it holds that the curvature $L_k \kappa(\gamma_k(s)) := \gamma_k''(s) \cdot (\gamma_k'(s))^{\perp} = L_k\theta'(s)$ is defined a.e $s \in [0,1]$ and is in $L^1([0,1])$ with
$\kappa(\gamma_k(s)) + \phi_{\Omega}(\gamma_k(s)) = \lambda$ holding for a.e $s \in [0,1]$. Then by standard elliptic theory
\cite{gilbarg}, it follows that $\gamma_k \in C^{3,\alpha}([0,1])$ for $\alpha \in (0,1)$, implying $C^{3,\alpha}$ regularity
of the boundary and that $\eqref{realel}$ holds strongly for all $s \in [0,1]$. 
\end{proof}
%that $\partial \Omega$ admits a $C^{3,\alpha}$ parametrization.
\noindent We then have the following approximation Theorem.
\begin{proposition}\label{approx} Let $\gamma$ be a closed rectifiable curve with $\theta \in BV([0,1])$. Then there exists a sequence
of $C^2$ curves $\gamma_n$ such that
\begin{align}
\gamma_n \to \gamma &\in W^{1,1}([0,1])\\
\theta_n \to \theta &\in L^1([0,1])\\
\theta_n' \rightharpoonup \theta' &\in (C([0,1])^*,
\end{align}
where $(C[0,1])^*$ is the dual of the space of continuous functions on $[0,1]$. 
\end{proposition}
\begin{proof}
Let $\theta_n := \eta^{1/n} * \theta$ where $\eta^{1/n} := \eta(nx)$ and $\eta$ is the standard mollifier. Then since $\|\theta\|_{BV([0,1])} < +\infty$ we have
\begin{equation}
\limsup_{n \to +\infty}\|\theta_n\|_{BV([0,1])}  < +\infty.
\end{equation}
Thus we have $\theta_n' \rightharpoonup \theta'$ in the weak sense of measures and $\theta_n \to \theta$ in $L^1([0,1])$ by
the embedding $BV([0,1]) \subset \subset L^1([0,1])$. The convergence $\gamma_n \to \gamma$ in $W^{1,1}([0,1])$
follows immediately.
\end{proof}
Using the above approximation Theorem we prove that the Gauss-Bonnet theorem continues to hold for closed, rectifiable curves. This is
not technically necessary in this section as we have proven that $\partial \Omega$ is always parameterizable by $C^{3,\alpha}$ curves by
Proposition \ref{weakcomp}. However we will use this result in Section \ref{stability} to prove the inequalities hold without the assumption
of smoothness of the boundary.
\begin{proposition}\label{gaussbonnet}
Let $\gamma$ be a closed, rectifiable curve with $\theta \in BV([0,1])$. Then there exists $N \in \mathbb{Z}$ such that 
\[\int_0^1 \theta'(s) ds = \int_{\partial \Omega} \kappa (y)dS(y) = 2\pi N,\]
where $N$ is called the winding number.
\end{proposition}
\begin{proof}
 Let $\gamma_n$ be as in Proposition \ref{approx}. It follows from the Gauss-Bonnet Theorem for $C^2$ curves that
\[ \int_0^{1} \theta_n'(s) ds = 2\pi N_n,\]
for all $n$ where $N_n \in \mathbb{Z}$ must be bounded uniformly in $n$, since $\theta_n$ is uniformly bounded in $BV([0,1])$. Using Proposition \ref{approx} we have $\theta_n' \rightharpoonup \theta'$ weakly in $(C([0,1])^*$ allowing us
to pass to the limit in the above, implying that $N_n=N$ for some $N \in \mathbb{Z}$ for sufficiently large $n$. Thus $\int_0^1 \theta'(s)ds = \int_{\partial \Omega} \kappa(y)dS(y) = 2\pi N$.
\end{proof}

\begin{proposition}\label{cpareconvex}
 Let $\Omega$ be a critical point of \eqref{energy3} in the sense of Definition \ref{defcp}, admissible
in the sense of Definition \ref{addef}. Then there exists
$\bar \eta _{cr} > 0$ such that whenever $\bar \eta  < \bar \eta _{cr}$, $\Omega$ is convex.
\end{proposition}
\begin{proof}
Each $\gamma_k \in C^{3,\alpha}([0,1])$ by Proposition \ref{weakcomp}. Let $\gamma_k$ be any interior curve parameterizing $\partial \Omega_k$, ie. the interior of $\gamma_k$ is contained in the interior of some other curve $\gamma_j$ for $j \neq k$. By Gauss-Bonnet (cf. Proposition \ref{gaussbonnet}) $\dashint_{\partial \Omega_k} \kappa (y)dS(y) = \frac{2\pi N}{L_k}$ where $N$ is the winding number of the curve $\gamma_k$, and \eqref{realel} we have
\begin{equation}\label{EL1}
 \kappa(y) + \phi_{\Omega}(y) = \frac{2\pi N}{L_k} + \bar \phi_{\Omega}
\end{equation}
holds for $y \in \partial \Omega_k$ where the bar denotes average over $\partial \Omega_k$ and where $N \in \mathbb{Z}$. We wish
to show that $\gamma_k$ is a simple curve, ie. $N=-1$. 
There is a universal constant $C>0$ such that
\begin{align}\label{phibound}
 |\phi_{\Omega}(y)| &\leq C m(1+|\log L|) \textrm{ when } G(x,y) = -\frac{1}{2\pi} \log |x-y|\\
|\phi_{\Omega}(y)| &\leq C  \frac{m}{L^{\alpha}} \textrm{ when } G(x,y) = \frac{1}{|x-y|^{\alpha}}.\nonumber
\end{align}
Assume first that $N>0$. Then for $y \in \partial \Omega_k$ we deduce from \eqref{EL1}--\eqref{phibound} and $L_k \leq L$
\begin{align}
 \label{kineq1a} \kappa(y) &\geq\frac{1}{L} \( 2\pi N - C m L (|\log L|+1)\) \geq \frac{1}{L}(2\pi N - 2\sqrt{\pi} C \bar \eta ) \textrm{ when } G(x,y)=-\frac{1}{2\pi}\log|x-y|\\
\label{kineq2a}\kappa(y) &\geq \frac{1}{L}\(2\pi N - C m L^{1-\alpha}\) \geq \frac{1}{L}(2\pi N - 2\sqrt{\pi} C \bar \eta ) \textrm{ when } G(x,y)=\frac{1}{|x-y|^{\alpha}}.
\end{align}
where the second inequalities follow from the isoperimetric inequality $2\sqrt{\pi} m^{1/2} \leq L$. It is clear that when $\bar \eta $ is sufficiently small, $\kappa > 0$ for
all points on $\partial \Omega_k$, which is a clear contradiction since $\gamma_k$ was assumed to be an interior curve. When $N=0$ then we have once again from \eqref{EL1} and \eqref{phibound}
\begin{equation}\label{smallka} |L \kappa(y)| \leq C \bar \eta ,\end{equation}
for $C>0$ and all $y \in \partial \Omega_k$. Clearly there is always some $y \in \partial \Omega_k$ such that $\kappa(y) \geq \frac{\pi}{L_k}$. Indeed letting
$\gamma_k$ be a unit speed parametrization of $\partial \Omega_k$, we
restrict to $s \in [0,s_0] \subset [0,L_k]$ so that $0 \leq \theta_k(s) \leq 2\pi$. Then $\int_0^{s_1}  \theta_k'(s)ds = 2\pi$ and thus by the mean value
theorem, there exists an $s \in [0,s_0]$ such that $\theta_k'(s) = \frac{2\pi}{s_0} \geq \frac{\pi}{L_k}$ since $s_0 \in [0,L_k]$. 
 This contradicts \eqref{smallka} for $\bar \eta $ sufficiently small.
Arguing similarly when $N < 0$, we conclude that in this case, $\kappa < 0$ everywhere when $\bar \eta $ is sufficiently small and hence $\gamma_k$ is simple. Thus we have shown that each $\gamma_k$ is simple when
$\bar \eta $ sufficiently small since our estimates do not depend on $\gamma_k$. We now show that $\Omega$ must in fact be convex.

\noindent As before we have
\begin{equation}
 \kappa(y) + \phi_{\Omega}(y) = \frac{2\pi N}{L} + \bar \phi_{\Omega},
\end{equation}
where now the average is taken over all of $\partial \Omega$. Since each $\gamma_k$ is simple, $N \leq 1$ and we claim that in fact $N=1$.
First we show that $N \neq 0$. In  this case we have $\bar \kappa=0$ and thus from \eqref{EL1} and \eqref{phibound} 
\begin{equation}\label{smallk} |L \kappa(y)| \leq C \bar \eta ,\end{equation}
for $C>0$ and all $y \in \partial \Omega$. As before, there is always some $y \in \partial \Omega$ such that $\kappa(y) \geq \frac{\pi}{L_1}$,
where $L_1$ denotes the length of outer component of $\partial \Omega$, denoted as $ \partial \Omega_1$ (ie. the interior of $\gamma_k$ is contained
in the interior of $\gamma_1$ for all $k$). This is a contradiction
of \eqref{smallk} however since $L_1 \leq L$. To see that $N \geq 0$, assume that $N<0$. Then once again from \eqref{EL1} and \eqref{phibound}  
\begin{align}
 \label{kineq1} \kappa(y) &\leq\frac{1}{L} \( -2\pi N + C m L (|\log L|+1)\) \leq \frac{1}{L}(-2\pi N + 2\sqrt{\pi} C \bar \eta ) \textrm{ when } G(x,y)=-\frac{1}{2\pi}\log|x-y|\\
\label{kineq2}\kappa(y) &\leq \frac{1}{L}\(-2\pi N + C m L^{1-\alpha}\) \leq \frac{1}{L}(-2\pi N + 2\sqrt{\pi} C \bar \eta ) \textrm{ when } G(x,y)=\frac{1}{|x-y|^{\alpha}}.
\end{align}
By choosing $\bar \eta $ sufficiently small, then we would have $\kappa < 0$ everywhere on $\partial \Omega_1$ which is a contradiction, since $\gamma_1$ was assumed to be the exterior curve. Thus $N=1$ and $\Omega$ is simply connected, ie. $\gamma_1$ is simple.
Then we have by Proposition \ref{gaussbonnet} and \eqref{realel} again that
\begin{align}
 \label{kineq1a} \kappa(y) &\geq \frac{1}{L} \( 2\pi  - C \bar \eta \) \textrm{ when } G(x,y)=-\frac{1}{2\pi}\log|x-y|\\
\label{kineq2b}\kappa(y) &\geq \frac{1}{L}\(2\pi - C \bar \eta \) \textrm{ when } G(x,y)=\frac{1}{|x-y|^{\alpha}},
\end{align}
for all $y \in \partial \Omega$, showing that $\kappa > 0$ whenever $\bar \eta $ is chosen small enough. Thus $\Omega$ is convex when $\bar \eta $ is chosen sufficiently small.
\end{proof}
The assumption that $\bar \eta $ be sufficiently small is not simply a technical assumption, as Counter Example \ref{CE1} demonstrates.

\section{Area preserving mean curvature flow}\label{VPMCFsec}

We let $\Omega$ be a smooth, compact subset of $\mathbb{R}^2$ with boundary $\partial \Omega$. Letting $X_0$
be a local chart of $\Omega$ so that
\[ X_0: E \subset \mathbb{R}^{2} \to X_0(E) \subset \partial \Omega \subset \mathbb{R}^2.\]
Then we let $X(x,t)$ be the solution to the evolution problem
\begin{align}\label{VPMCF}
 \frac{\partial}{\partial_t} X(x, t) &= - (\kappa(t,x)- \bar \kappa(t)) \cdot \nu(x,t), \;\; x \in E, t \geq 0\\
\nonumber X(\cdot,0) &= X_0,
\end{align}
where $\kappa(t,x)$ is the mean curvature of $\partial \Omega$ at the point $x$, $\bar \kappa(t)$ is the average of the mean
curvature on $\partial \Omega_t$:
\begin{equation} \label{SAdec} \bar \kappa(t) = \dashint_{\partial \Omega_t} \kappa(y)dS(y),\end{equation}
$\nu(x,t)$ is the normal to $\partial \Omega_t$ at the point $x$ and $dS$ is the 1-dimensional Hausdorff measure. The flow \eqref{VPMCF} in dimensions $d \geq 3$ was first introduced by Huisken \cite{huisken} who established
existence and asymptotic convergence to round spheres for initially convex domains. The planar version for curves was introduced by Gage \cite{Gage}. \\

\noindent For convenience of notation we define
\begin{align}
 L &:= |\partial \Omega| \label{length}\\
A &:= |\Omega|.\label{area}
\end{align}
 It is easy to see that the surface area of $\partial \Omega$ is decreasing along the flow. Indeed differentiating the perimeter
we have
\begin{equation} \label{areadecrease} \frac{dL}{dt} = - \int_{\partial \Omega_t} (\kappa - \bar \kappa)^2 dS.\end{equation}
The introduction of the non-local term in \eqref{energy3} will create a term which competes with \eqref{areadecrease} along the flow \eqref{VPMCF} as the non-locality
favors the spreading of mass. The main element of the proof of Theorem \ref{main4} stated in Section \ref{mainresults} will therefore be to show that
when the mass is small, the decay in perimeter is sufficient to compensate for the increase in energy of the non-local term in \eqref{energy3} along
the flow \eqref{VPMCF}. This is where Theorem \ref{nonexistI} will play a crucial role. Before we proceed we must recall some now well known results about the flow \eqref{VPMCF}.

\medskip
\noindent  The main result of \cite{Gage} due to Gage
is the following:
\begin{theorem}(Global existence)\label{Huiskenexist}
 If $\Omega$ is convex, then the evolution equation \eqref{VPMCF} has a smooth solution $\Omega_t$
for all times $0 \leq t < \infty$ and the sets $\Omega_t$ converge in the $C^{\infty}$ topology to a round sphere, enclosing the same volume as $\Omega$, exponentially fast, as $t \to +\infty$. 
\end{theorem}
\noindent In addition there is a local in time existence result in \cite{Gage}
\begin{theorem}(Local existence)\label{Huiskenexist2}
 If $\Omega$ is any smooth embedded set, then there exists a $T>0$ such that the evolution equation \eqref{VPMCF} has a smooth solution $\Omega_t$
for $t \in [0,T)$.
\end{theorem}

\noindent Finally we are able to state the main result concerning the flow \eqref{VPMCF} and the energy \eqref{energy3}.

\section{Main Results}\label{mainresults}

Our main result, from which the other results follow, is the following.
\begin{theorem} (Energy decrease along the flow) \label{main4} Let $\Omega$ be admissible in the sense of Definition \ref{addef} with $\kappa \in L^2(\partial \Omega)$ and denote $\Omega_t$ the evolution of $\Omega$ under \eqref{VPMCF}.
\begin{itemize}
 \item If $U = \mathbb{R}^2$ and $\Omega$ is convex, then there exists $\bar \eta _{cr} > 0$ such that whenever $\bar \eta  < \bar \eta _{cr}$ it holds that
\begin{equation}\label{energydec}
 \frac{dE(\Omega_t)}{dt}\big|_{t=0} < 0,
\end{equation}
with $E$ defined by \eqref{energy3}.
\item If $U=\mathbb{T}^2$ and $\Omega$ is either convex or homeomorphic to $S_{n=1}$ then there
 exists an $\bar \gamma_{cr}=\bar\gamma_{cr} > 0$ such that whenever $\bar \gamma < \bar \gamma_{cr}$ \eqref{energydec} holds.
\end{itemize}
 Moreover lower bounds for $\bar \eta _{cr}$ and $\bar \gamma_{cr}$ are given by 
\[ \max_{\Omega \subset \mathbb{R}^2} \(\bar \eta ^2\frac{\int_{\partial \Omega}(\kappa - \bar \kappa)^2 dS}{ L \|\phi - \bar \phi\|_{L^{\infty}(\partial \Omega)}}\)^{\frac{1}{2}} \textrm{ and }\;\; \max_{\Omega \subset \mathbb{T}^2} \(\bar \gamma^2\frac{\int_{\partial \Omega}(\kappa - \bar \kappa)^2 dS}{ L \|\phi - \bar \phi\|_{L^{\infty}(\partial \Omega)}}\)^{\frac{1}{2}} \]
respectively, where the maximum is taken over convex sets for the first expression, and over convex sets and sets homeomorphic to $S_{n=1}$ in the latter.

\begin{itemize} \item If $\Omega$ is in addition assumed to be smooth, then there exists a constant a $T>0$ and $C>0$ depending only on $\bar \eta _{cr}$ (or $\bar \gamma_{cr}$) such that
\begin{equation}\label{Edecay2}
 \frac{dE(\Omega_t)}{dt} \leq -C \int_{\partial \Omega_t} (\kappa - \bar \kappa)^2 dS_{\Omega_t}
\end{equation}
for all $t \in [0,T)$. When $U=\mathbb{R}^2$ and $\Omega$ is convex, $T=+\infty$. 
\end{itemize}
\end{theorem}

\begin{remark}
From Theorem \ref{nonexist} below, we will see that a lower bound for $\bar \eta _{cr}$ is given by $\frac{32}{\pi}$ for $G(x,y)=-\frac{1}{2\pi} \log |x-y|$
and $\frac{8}{\pi}\(1 + \frac{1}{\pi}\)^{1-\alpha}$ for $G(x,y)=\frac{1}{|x-y|^{\alpha}}$. Also a lower bound for $\bar \gamma_{cr}$
when $U=\mathbb{T}^2$ is given by a universal constant $C_0>0$ depending only on the Green's function for $\mathbb{T}^2$. 
\end{remark}

%\begin{counter}\label{CE2}
% When $U=\mathbb{R}^2$ then there exists an annuli $\Omega := \{ x : r < |x| < R\}$ such that 
%\[ \frac{dE(\Omega_t)}{dt}\big|_{t=0} = C(\bar \eta  - c),\]
%for constants $c=c(\alpha)$, $C=C(\alpha) > 0$. In particular, there exists a $\bar \eta _{cr,2} > \bar \eta _{cr}$ such that Theorem
%\ref{main4} is false for $\bar \eta  > \bar \eta _{cr,2}$. 
%\end{counter}
\medskip
\noindent Using this result we can immediately prove Corollary \ref{R2stable}.\\

\noindent {\sc Proof of Corollary \ref{R2stable}}
Using Proposition \ref{weakcomp}, we can invoke
\eqref{Edecay2} along with \eqref{areadecrease} to conclude the desired inequality. 
 $\Box$

\section{Geometric inequalities}\label{stability}
 
The main result of this section is a geometric inequality which relates the closeness of connected constant curvature curves and equipotential
curves. This will be used in Section \ref{finalproof} to control the non-local term along the flow \eqref{VPMCF}. 

\begin{theorem}\label{nonexist}

Let $\Omega \subset U$ be convex with $\kappa \in L^2(\partial \Omega)$ and $U=\mathbb{R}^2$ or $U=\mathbb{T}^2$. Then there is a constant $C=C(L,\alpha)>0$ such that
\begin{equation}\label{gineq}\|\phi_{\Omega} - \bar \phi_{\Omega}\|_{L^{\infty}(\partial \Omega)}^2 \leq C \int_{\partial \Omega}|\kappa - \bar \kappa|^2dS.\end{equation}  
\begin{itemize}
\item[i)] When $U=\mathbb{R}^2$ and $G(x,y)=-\frac{1}{2\pi} \log |x-y|$, the constant can be chosen to be $C=C_{\mathbb{R}^2}=\frac{32AL^3}{\pi} (1+|\log L|)^2 = \frac{32}{\pi L} \bar \eta ^2$.
\item[ii)] When $U=\mathbb{R}^2$ and $G(x,y)=\frac{1}{|x-y|^{\alpha}}$ for $\alpha \in (0,1)$, the constant can be chosen to be
$C=\frac{8A}{\pi} \(1+\frac{1}{\pi}\)^{2(1-\alpha)} L^{3-2\alpha}=\frac{8}{\pi L} \bar \eta ^2$.
\item[iii)] When $U=\mathbb{T}^2$ the constant
can be chosen to be $C=C_{\mathbb{T}^2} = C_0 L^3(1+ |\log L|^2)^2 = \frac{C_0}{\gamma L} \bar \gamma^2 $ where $C_0>0$ is some universal constant. Moreover the inequality continues to hold
for sets homeomorphic to $S_{n=1}$.
\end{itemize}                                                                                            
\end{theorem}

We separate the case of $U=\mathbb{R}^2$ and $U=\mathbb{T}^2$ since the ability to estimate the distance to an equipotential surface
in terms of the distance to a constant curvature surface (with some abuse of terminology hereafter denoted CMC surface) depends on the types of CMC surfaces which may exist. In $\mathbb{R}^2$ we first
control the isoperimetric deficit in terms of the right hand side of \eqref{gineq}, then control the left hand side of \eqref{gineq} by the
isoperimetric deficit. This relies crucially on the fact that the ball is the only compact, connected CMC surface in $\mathbb{R}^2$. In
$\mathbb{T}^2$ we must account for the possibility of the stripe pattern defined by \eqref{R.6} and need to adapt the inequalities accordingly. Therefore
while the approaches are almost identical, separate analysis is needed. The Euclidean case will be treated below in Section \ref{Rn} and the case
of the torus will be treated in Section \ref{Torus}. We will make repeated use of Bonnesen's inequality \cite{bonnesen}, which states that given
any set simply connected $\Omega \subset \mathbb{R}^2$, it holds that
\begin{equation} \label{strong2}
 L^2 - 4\pi A \geq \pi^2 (R_{\textrm{out}}(\Omega)-R_{\textrm{in}}(\Omega))^2,
\end{equation}
where $R_{\textrm{in}}$ and $R_{\textrm{out}}$ are
\begin{align}\label{inout}
 R_{\textrm{in}}(\Omega) &= \sup_{B_r \subset \Omega} r\\
R_{\textrm{out}}(\Omega) &= \inf_{\Omega \subset B_R} R.\nonumber
\end{align}
See Figure 1 for a diagram showing the various quantities. 
\subsection{The Euclidean case $U=\mathbb{R}^2$}\label{Rn}
Let $X(s)=(x(s),y(s)) \in \mathbb{R}^2$ be a unit speed parametrization of the curve enclosing $\Omega$, $\nu$ the normal to the surface at the point $X(s)$ and $p(s) := \langle X(s), -\nu(s) \rangle$ the support function.
We will also use $A$ instead of $m$ for the area to emphasize the geometric nature of the inequalities, although the two are equivalent.

\begin{figure}
 \centering

\includegraphics[scale=0.3]{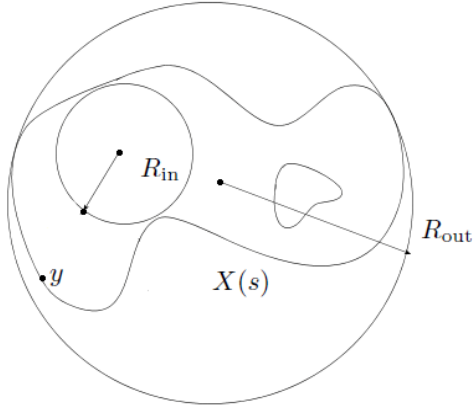}
\caption[Test]{The boundary of $\Omega$ is parametrized by $X(s)$ and has inner and outer radii $R_{in}$ and $R_{out}$.}
\end{figure}
\begin{proposition}\label{isopermineq} Let $\Omega \subset \mathbb{R}^2$ be convex with $\kappa \in L^2(\partial \Omega)$. Then
it holds that
 \begin{equation}\label{isoineq2}
  L - 2\sqrt{\pi} A^{1/2} \leq \frac{A}{\pi} \int_{\partial \Omega}( \kappa - \bar \kappa)^2. 
 \end{equation}

\end{proposition}
\begin{proof}
Using the generalized Gauss-Green theorem we have 
\begin{equation}
 A = \int_{\Omega} \det DX(x,y) dxdy  = \frac{1}{2}\int_{\partial \Omega} \langle X, - \nu \rangle dS = \frac{1}{2} \int_0^L p(s) ds,
\end{equation}
where we set $p(s)=\langle X(s), \nu(s)$ and $\kappa(s) = X''(s) \cdot \nu(s)$. In addition we have
\begin{align}
 \int_{\partial \Omega} p \kappa dS = \int_0^L p(s) \kappa(s) ds &= -\int_0^L \langle X(s), \kappa \nu(s) \rangle ds = -\int_0^L \langle X(s), X''(s) \rangle ds \nonumber \\
&= - \langle X, X' \rangle \big|_0^L + \int_0^L \langle X'(s), X'(s) \rangle ds \nonumber \\
&= L.\label{GageCalc}
\end{align}
Adding and subtracting $\bar \kappa$ to $\kappa$ in the integrand on the left side of \eqref{GageCalc}, and using the Gauss-Bonnet theorem for curves (cf. Proposition \ref{gaussbonnet}), we find
\begin{equation}\label{SCeqn}
 L - \frac{4\pi A}{L} = \int_{\partial \Omega} p (\kappa - \bar \kappa) dS.
\end{equation}

We first prove the inequality when $\Omega$ is symmetric about the origin of $p$. Adding and subtracting $\bar p$ from $p$ in \eqref{SCeqn} we have
\begin{align}\label{est4}
\frac{L^2-4\pi A}{L} &\leq \int_{\partial \Omega} (p - \bar p) (\kappa - \bar \kappa) dS\\
&\leq \sqrt{ \int_{\partial \Omega} (p- \bar p)^2 dS} \sqrt{ \int_{\partial \Omega} (\kappa - \bar \kappa)^2 dS},\nonumber
\end{align}
where we've applied Cauchy-Schwarz on the last line. A simple calculation yields
\begin{equation}\label{pfunc}
 \int (p - \bar p)^2 = \int p^2 - \frac{4A^2}{L}
\end{equation}
Using an inequality due to Gage \cite{Gage2} for convex sets symmetric about the origin, we have 
\begin{equation}\label{gageineq}
 \int p^2 dS \leq \frac{LA}{\pi}.
\end{equation}
Inserting \eqref{gageineq} into \eqref{pfunc} we have
\begin{equation}
 \int (p-\bar p)^2 dS \leq \frac{A}{\pi L} \(L^2 - 4\pi A\).
\end{equation}

\noindent Inserting the above into \eqref{est4} we have
\begin{align} \sqrt{L^2-4\pi A} &\leq \sqrt{\frac{AL}{\pi}} \sqrt{ \int_{\partial \Omega} (\kappa - \bar \kappa)^2 dS}.
\end{align}
 Squaring both sides, dividing both sides by $L$ and using $L \leq 2\sqrt{\pi} A^{1/2}$ we have
\begin{align} L - 2\sqrt{\pi} A^{1/2} &\leq \frac{A}{\pi} \int_{\partial \Omega} (\kappa - \bar \kappa)^2 dS.
\end{align}
When $\Omega$ is not symmetric about the origin, then choose a point $O \in \Omega$ and any straight line passing through $O$. Then this line will divide the set $\Omega$ into two segments, $\Omega_1$ and $\Omega_2$. We claim it is always possible to choose this line so that $\Omega_1$ and $\Omega_2$ have the same area. Indeed
let $F(\Omega, \theta) := |\Omega_1| - |\Omega_2|$ where $\theta$ denotes the angle of the line with respect to some fixed axis. If $F(\Omega,0)=0$ then we are done. Otherwise $F(\Omega,0) > 0$ without loss of generality. But then $F(\Omega,2\pi) < 0$ and hence by the intermediate value theorem we conclude there exists a $\theta \in (0,2\pi)$ such that $F(\Omega,\theta)=0$. Without loss of generality we may orient this line to be
 parallel to the x axis. Let $\Omega_1'$ and $\Omega_2'$ be the sets formed by reflection across
the origin. Then we can apply \eqref{isoineq2} and we obtain for $i=1,2$
\[ L_i' - 2 \sqrt{\pi }A^{1/2} \leq \frac{A}{\pi}\int_{\partial \Omega_i'} (\kappa_i' - \bar \kappa_i')^2 dS = \frac{A}{\pi} \(\int_{\partial \Omega_i'} (\kappa_i')^2 - \frac{4\pi^2}{L_i'} \).\]
Adding the two inequalities over $i=1, 2$ and using
\[ \frac{1}{L_1'} + \frac{1}{L_2'} \geq \frac{2}{L_1' + L_2'},\]
we obtain after division by $2$
\[ L - 2 \sqrt{ \pi } A^{1/2} \leq \frac{ A}{\pi}\int_{\partial \Omega} \kappa^2 dS - \frac{4\pi^2}{L}= \frac{ A}{\pi}\int_{\partial \Omega}(\kappa - \bar \kappa)^2 dS,\]
the desired inequality.
\end{proof}

Next we obtain a quantitative estimate for the closeness to
an equipotential surface in terms of the isoperimetric deficit. We will present the proof below for $G(x,y)=-\frac{1}{2\pi} \log |x-y|$
and show how the proof is adapted to the case $G(x,y)=\frac{1}{|x-y|^{\alpha}}$ in the remark following the proof.

\begin{proposition}\label{equicontrol} Let $\Omega \subset \mathbb{R}^2$ be simply connected. Then there exists a constant $C=C(L,\alpha)>0$ such that
 \begin{align}
 \|\phi_{\Omega} - \bar \phi_{\Omega}\|_{L^{\infty}(\partial \Omega)}^2 \leq C (L^2 -4\pi A),
\end{align}
where $\phi_{\Omega}$ is defined by \eqref{potdefn1} for $G$ defined by \eqref{kerndef1}.
\begin{itemize}
 \item When $G(x,y) = -\frac{1}{2\pi} \log |x-y|$ the constant can be chosen to be $C=16 L^2(1+|\log L|)^2$.
\item When $G(x,y) = K(x-y)$ for $\alpha \in (0,1)$, the constant can be chosen to be $C=4\(1+\frac{1}{\pi}\)^{2-2\alpha} L^{2-2\alpha} $. 
\end{itemize}

\end{proposition}

\noindent {\sc Proof of Proposition \ref{nonexist}}

\noindent Consider any two points $y, z \in \partial \Omega$  and assume first $\phi_{\Omega}(y) > \phi_{\Omega}(z)$. Then we have
\begin{align}
 \phi_{\Omega}(y) - \phi_{\Omega}(z) &= \int_{\Omega} (G(x,y) + C) - (G(x,z)+C) dx \\
&\leq \int_{B_{R_{\textrm{out}}}}G(x,y)+C  - \int_{B_{R_{\textrm{in}}}} G(x,z)+C,\label{radialeqn}
\end{align}
where $C=C(R_{\textrm{out}})>0$ is a constant chosen so that $G+C$ is positive on $B_{R_{\textrm{out}}}$. In particular $C$
can be chosen to be
\[ C = C_1 := \max (0, \frac{1}{2\pi} \log 2 R_{\textrm{out}}).\]
\noindent Using radial symmetry of the Laplacian equation \eqref{radialeqn} is in fact equal to 
\begin{equation}\label{lapsols}
 \frac{1}{2} (R_{\textrm{out}}^2 - |y|^2) - \frac{1}{2}(|z|^2 - R_{\textrm{in}}^2) + C_1 (R_{\textrm{out}}^2 - R_{\textrm{in}}^2) - \frac{R_{\textrm{out}}^2}{2} \log R_{\textrm{out}} + \frac{R_{\textrm{in}}^2}{2} \log R_{\textrm{in}}
\end{equation}
Arguing similarly when $\phi_{\Omega}(y) \leq \phi_{\Omega}(z)$ we conclude for all $y, z \in \partial \Omega$ that
\begin{align}\label{lastest1}
 |\phi_{\Omega}(y) - \phi_{\Omega}(z)| &\leq (C_1+2R_{\textrm{out}}) \(R_{\textrm{out}}-R_{\textrm{in}}\) + \frac{1}{2}\left| R_{\textrm{out}}^2 \log R_{\textrm{out}} - R_{\textrm{in}}^2 \log R_{\textrm{in}}\right|\\
&\leq \(2 + \frac{1}{\pi}\)R_{\textrm{out}}\(R_{\textrm{out}}-R_{\textrm{in}}\)  + \frac{1}{2}\left| R_{\textrm{out}}^2 \log R_{\textrm{out}} - R_{\textrm{in}}^2 \log R_{\textrm{in}}\right|.\nonumber
\end{align}
 Observe that the function $f(x) = x^2 \log x$ with $f(0)=0$ is $C^1$ and so assuming $x > y$ we have
\begin{equation}
 |f(x) - f(y)| \leq |f'|_{L^{\infty}(0, x)} |x-y| \leq 2|x|(1+|\log x|) |x-y|,
\end{equation}
where in the last inequality we've used the fact that $x \mapsto 2|x|(1+|\log x|)$ is monotone increasing on $(0,+\infty)$. Using Bonnesen's inequality (cf. equation \eqref{strong2})
we have 
\begin{align}
 R_{\textrm{out}} \leq R_{\textrm{in}} + \frac{1}{\pi} \sqrt{L^2 - 4\pi A} \leq \frac{2L}{\pi}
\end{align}
and thus
\begin{equation}
 \left| R_{\textrm{out}}^2 \log R_{\textrm{out}} - R_{\textrm{in}}^2 \log R_{\textrm{in}}\right| \leq 2L ( |\log L| + 1) (R_{\textrm{out}} - R_{\textrm{in}}) \leq 2L(|\log L| + 1) \sqrt{L^2-4\pi A}
\end{equation}

 \noindent Inserting this into \eqref{lastest1} and using \eqref{strong2}, we obtain
\begin{equation}
 |\phi_{\Omega}(y) - \phi_{\Omega}(z)| \leq 4L(1+|\log L|)\sqrt{L^2-4\pi A}.
\end{equation}
Choosing $z$ so that $|\phi_{\Omega}(y) - \phi_{\Omega}(z)| \leq |\phi_{\Omega}(y) - \bar \phi|$ and maximizing over $y \in \partial \Omega$ yields the desired inequality.
 $\Box$
 \begin{remark} The above proof easily adapts to the case of the kernel $K(x) = \frac{1}{|x|^{\alpha}}$. Indeed the constant $C_1$ above can be taken to be zero, since $K > 0$ everywhere and line \eqref{lapsols} can simply be replaced by 
\[ \(\int_{B(0,1)} \frac{dx}{|x-y|^{\alpha}}\)\( R_{\textrm{out}}^{2-\alpha} - R_{\textrm{in}}^{2-\alpha}\) \leq 2\pi R_{\textrm{out}}^{1-\alpha} (R_{\textrm{out}} - R_{\textrm{in}}) \textrm{ for } x \in \partial B(0,1),\]
where we've performed a first order Taylor expansion of the function $f(x)=x^{2-\alpha}$ about the point $x=R_{\textrm{in}}$. Using Bonnesen's inequality once again we obtain
\[|\phi_{\Omega}(y) - \phi_{\Omega}(z)| \leq 2 \(1 + \frac{1}{\pi}\)^{1-\alpha} L^{1-\alpha}\sqrt{(L^2 - 4\pi A)}.\]
The rest of the proof is argued identically to the logarithmic case. 
\end{remark}
 \noindent {\sc Proof of Theorem \ref{nonexist} item i)}
Now the proof of Theorem \ref{nonexist}, items $i)$ and $ii)$ immediately follows from combining Proposition \ref{equicontrol} and Proposition \ref{isopermineq} and using $L^2-4\pi A \leq 2L (L-2\sqrt{\pi} A^{1/2})$.

\subsection{The torus case $U=\mathbb{T}^2$}\label{Torus}
We restrict to simply connected sets in $\mathbb{T}^2=[0,1)^2$. Any simply connected set lying in $\mathbb{T}^2$ must either
be homeomorphic to the disk or to $S_{n=1}$ defined by \eqref{R.6}. In the former case, the analysis is identical to that of the $U=\mathbb{R}^2$ case above. In the latter case, we proceed as follows. By possibly changing coordinates we can see $\Omega$ as represented by a curve $X$ with two separate components, $X_{\pm}:[0,L^{\pm}) \to \mathbb{T}^2$, where $L^{\pm}$ denote the lengths of $X_{\pm}$, which are unit
speed parameterizations of the two components of $\partial \Omega$. We denote $S_L$ as a rectangle of width $L$ and height 1 in the torus. Then we define
\begin{align}
 L_{\textrm{out}} &:= \min \{L : \Omega \subset S_L\},\\
L_{\textrm{in}} &:= \max\{ L : S_L \subset \Omega\}.
\end{align}
When $X$ is $C^1$, we denote $S_{L_{\textrm{in}}}$ and $S_{L_{\textrm{out}}}$ as the two rectangles with widths $L_{\textrm{in}}$ and $L_{\textrm{out}}$ respectively such
that $X_{\pm}$ lies in between these two rectangles and touches them tangentially (see Figure 2). The main difficulty lies in controlling the support function $p$ as was done in the
proof of Theorem \ref{isopermineq} via Gage's inequality (cf. equation \eqref{gageineq}). In this case we prove an inequality for the support function when $\Omega$
is homeomorphic to $S_{n=1}$ and star shaped (cf. Lemma \ref{starlem}) and then proceed to show that all simply connected critical points must be star shaped (cf. Proposition \ref{starshaped}). Finally using
the estimate established in Lemma \ref{starlem}, we prove the analogue of Proposition \ref{isopermineq} (cf. Proposition \ref{isopermineqT}) for sets homeomorphic to $S_{n=1}$ and subsequently the analogue
of Proposition \ref{equicontrol} (cf. Proposition \ref{equicontrolT}). 

\begin{lem}\label{starlem}
Assume $X(s)$ is $C^1$ and star shaped with respect to the center of $S_{L_{\textrm{in}}}$. Let $p$ be the support function for $X(s)$ and $p^*$ the support function for $\partial S_{L_{\textrm{out}}}$ with respect to this center. Then there exists a universal constant $C>0$ such that
\[ \int_{\partial \Omega} p^2 dS \leq  \int_{\partial S_{L_{\textrm{out}}} } (p^*)^2 dS^* + C(L_{\textrm{out}}-L_{\textrm{in}}) + C(L-2),\]
where $dS$ denotes surface measure on $\partial \Omega$ and $dS^*$ the surface measure on $\partial S_{L_{\textrm{out}}}$.
\end{lem}

\begin{proof}
Let $O$ be the center of $S_{L_{\textrm{in}}}$, $s \in [0,L^{\pm}]$ the arc length parameter for $X_{\pm}$ and $\theta(s)$ the polar angle corresponding
to the point $X(s)$ on $\partial \Omega$. Since $\Omega$ is star shaped with respect to $O$, the mapping $s \mapsto \theta(s)$ is a bijection. Therefore we may set $p=p(\theta)$ for $\theta \in (\theta_1,\theta_2)$ and $(\theta_3,\theta_4)$ which parametrize $X_+(s)$ and $X_-(s)$ respectively. Consider $[0,1)^2$ as a subset of $\mathbb{R}^2$ and let
$X^*(s)$ be the projection of $X(s)$ onto $\partial S_{L_{\textrm{out}}}$. Observe that this projection will parametrize a subset of
$\partial S_{L_{\textrm{out}}}$ which includes the two vertical sides (See Figure 2) along with a segment of the horizontal pieces of $\partial S_{L_{\textrm{out}}}$
lying between $S_{L_{\textrm{in}}}$ and $S_{L_{\textrm{out}}}$ on $\partial [0,1)^2$, and our calculations below will include this segment. To avoid confusion we will denote $\tilde S_{L_{\textrm{out}}}$ and $\tilde S_{L_{\textrm{in}}}$ as
the entire rectangles, which includes the top segments along $\{(x,y) : 0 \leq x \leq 1, y = 0 \textrm{ or } y=1\} \subset \partial [0,1)^2$. 

Consider a small ray of angular width $\alpha \ll 1$ extending from $O$ and such that $\partial \Omega \neq \partial \Omega^*$ in this ray.
Define $R_{\alpha}$ to be the segment of the ray in between the curves $\partial S_{L_{\textrm{out}}}$ and $\partial \Omega$. Then from
Green's theorem we have
\begin{equation}
\int_{R_{\alpha}} X \cdot \nu dS_{R_{\alpha}} = |S_{L_{\textrm{out}}} \cap R_{\alpha}| - |\Omega \cap R_{\alpha}| := A_{\alpha}.
\end{equation}
Breaking up the integral on the right, we have
\begin{equation}
\int_{\partial \tilde S_{L_{\textrm{out}}} \cap R_{\alpha}} p^*(\theta) \frac{dS^*}{d\theta}d\theta = \int_{\partial \Omega \cap R_{\alpha}} p(\theta) \frac{dS}{d\theta} d\theta + A_{\alpha}.
\end{equation}
Now divide by $\alpha$ and send $\alpha \to 0$. Then we obtain
\begin{equation}
p^*(\theta) \frac{dS^*}{d\theta} = p(\theta) \frac{dS}{d\theta}  + \frac{1}{2}[[r^*(\theta)]^2 - r(\theta)^2],\end{equation}
where $r^*$ and $r$ denote the distances to $\partial \tilde S_{L_{\textrm{out}}}$ and $\partial \Omega$ respectively. 
Then multiplying by $p$ and $p^*$ we have
\begin{align}\label{L1.1}
p^*(\theta)p(\theta) \frac{dS^*}{d\theta} = p(\theta)^2 \frac{dS}{d\theta} + \frac{p(\theta)}{2} [[r^*(\theta)]^2-r(\theta)^2]
\end{align}
and 
\begin{align}\label{L1.2}
p(\theta)p^*(\theta) \frac{dS}{d\theta} = p^*(\theta)^2\frac{dS^*}{d\theta} + \frac{p^*(\theta)}{2}[[r^*(\theta)]^2 - r(\theta)^2],
\end{align}
respectively.
Integrating, we have
\begin{align}
\int_{\theta_1}^{\theta_2} p^2(\theta) \frac{dS}{d\theta} d\theta &\stackrel{\eqref{L1.1}}{=} \int_{\theta_1}^{\theta_2}p^*(\theta)p(\theta) \frac{dS^*}{d\theta}  d\theta - \int_{\theta_1}^{\theta_2}\frac{p(\theta)}{2}[[r^*(\theta)]^2-r(\theta)^2]d\theta \\
&\stackrel{\eqref{L1.2}}{=} \int_{\theta_1}^{\theta_2} p^*(\theta)^2\frac{dS^*}{d\theta} d\theta + \int_{\theta_1}^{\theta_2} p^*(\theta)p(\theta) \( \frac{dS^*}{d\theta} - \frac{dS}{d\theta}\)\\
&- \frac{1}{2}\int_{\theta_1}^{\theta_2}[p^*(\theta)-p(\theta)][[r^*(\theta)]^2 - r(\theta)^2] d\theta \nonumber \\
&\leq \label{L1.4} \int_{\theta_1}^{\theta_2} p^*(\theta)^2\frac{dS^*}{d\theta} d\theta + 2\max_{\theta \in (\theta_1,\theta_2))} p^*(\theta)p(\theta) \((L^+-1) + 2(L_{\textrm{out}}-L_{\textrm{in}}) \)\\ &+ \frac{1}{2}\max_{\theta \in (\theta_1,\theta_2)} [p^*(\theta)-p(\theta)](|S_{L_{\textrm{out}}}|-A),\nonumber
\end{align}
where $L^+$ is the length of $X_+$ and \eqref{L1.4} has counted the contribution of $\int_{\theta_1}^{\theta_2} p^*(\theta)p(\theta) \( \frac{dS^*}{d\theta} - \frac{dS}{d\theta}\)$ along
the vertical pieces of $ \partial \tilde S_{L_{\textrm{out}}}$ and the horizontal segment $\partial \tilde S_{L_{\textrm{out}}} \backslash \partial \tilde S_{L_{\textrm{in}}}$. 
Next we observe that
\begin{align}
  \int_{\theta_3}^{\theta_4} p^*(\theta)^2\frac{dS^*}{d\theta} d\theta + \int_{\theta_1}^{\theta_2} p^*(\theta)^2\frac{dS^*}{d\theta} d\theta &= \int_{\tilde \partial S_{L_{\textrm{out}}} \backslash \tilde \partial S_{L_{\textrm{in}}}} (p^*)^2dS^* + \int_{\partial S_{L_{\textrm{out}}}} (p^*)^2dS^*\\ &\leq 8(L_{\textrm{out}}-L_{\textrm{in}}) + \int_{\partial S_{L_{\textrm{out}}}} (p^*)^2 dS^*.
\end{align}

Repeating the above on the curve $X_-$, using the uniform boundendess of $p$ and $p^*$  and $|S_{L_{\textrm{out}}}- A| \leq |L_{\textrm{out}}-L_{\textrm{in}}|$ yields the result.
\end{proof}
\begin{proposition}\label{starshaped}
For $\bar \eta $ sufficiently small, any critical point on $\mathbb{T}^2$ is star shaped with with respect to the center of $S_{L_{\textrm{in}}}$.
\end{proposition}
\begin{proof}
Arguing as in Section \ref{weakeleqn} it is easily seen that the
curves $X_{\pm}$ are $C^{3,\alpha}$. Assume the result of the Proposition is false for $X_{\pm}(s)$ and let $O$ be the center of $S_{L_{\textrm{in}}}$.  Then there exists an angle $\theta \in (\theta_1,\theta_2)$ (without loss of generality
we choose the curve on the left, $X_+(s)$ in Figure 2) such that
$O +(\cos \theta,\sin \theta)$ intersects two points on $X_+(s)$, which we denote as $X_+(s_1)$ and $X_+(s_2)$. Then necessarily
$\nu(X_+(s_1)) = - \nu (X_+(s_2))$ for two pairs of such points. Thus, by the mean value theorem it follows that there exists an $s \in [s_1,s_2]$ such that
\[ \kappa(X_+(s)) = \frac{\pi}{s_2-s_1} \geq \frac{\pi}{L}.\]
Then using \eqref{realel} we ahve
\[ \kappa(X_+(s)) \geq \frac{1}{L} (\pi - L(1+|\log L|) \gamma C ),\]
where $C>0$ is universal. We thus have
\[ \kappa(X_+(s)) \geq \frac{1}{L} (\pi -  \bar \gamma),\]
 since $L \geq 1$. Choosing $\bar \gamma$ sufficiently small yields a contradiction of \eqref{realel} since $|\kappa|L \leq C \gamma L^2 (1+|\log L|) \leq C \bar \gamma$. 
\end{proof}

\noindent We then have the following version of Proposition \ref{isopermineq}. We present the proof in the case $A=\frac{1}{2}$ for
simplicity of presentation but the proof is easily adapted for any $A \in (0,1)$. 
\begin{figure}
 \centering
 \includegraphics[scale=0.3]{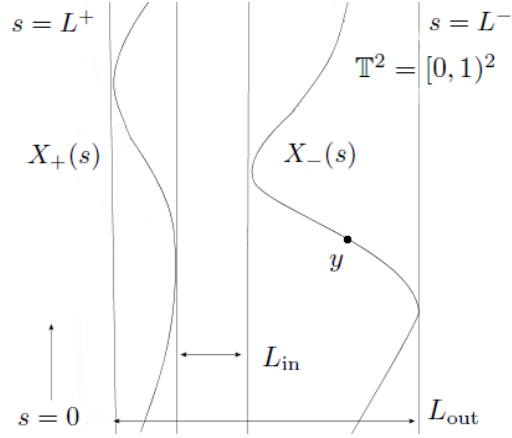}
 \caption{The curve $X(s)$ is composed of two components, $X_{\pm}(s)$ contained in between two rectangles $S_{L_{\textrm{out}}}$ and
 $S_{L_{\textrm{in}}}$ of lengths $L_{\textrm{out}}$ and $L_{\textrm{in}}$ respectively. }

 \end{figure}
\begin{proposition}\label{isopermineqT} Let $\Omega \subset \mathbb{T}^2$ have a $C^2$ boundary and be homeomorphic to $S_{n=1}$. Assume in addition that $\Omega$ is star shaped with respect to the center of $S_{L_{\textrm{in}}}$.  Then there is a universal constant $C>0$ such that
 \begin{equation}
  |L_{\textrm{out}}-L_{\textrm{in}}|^2 \leq C\int_{\partial \Omega}( \kappa - \bar \kappa)^2.
 \end{equation}
\end{proposition}
\begin{proof}
%Assume first that $X_{\pm}$ are $C^2$ curves. 
We perform all the analysis on one of the curves $X_+$ with length $L^+$, then argue symmetrically. We point out that below
averaged quantities such as $\bar p$ and $\bar \kappa$ will always denote the average over the entire curve $X$. As in Proposition \ref{isopermineq} we have
\begin{align}
 \int_0^{L^+} p \kappa dS &= -\int_0^{L^+} \langle X, \kappa \nu \rangle dS = -\int_0^{L^+} \langle X, X'' \rangle dS \nonumber \\
&= - \langle X_+, X_+' \rangle \big|_0^{L^+} + \int_0^{L^+} \langle X_+', X_+' \rangle dS \nonumber \\
&= L^+ - X_+'(L^+) \cdot X_+(L^+) + X_+'(0) \cdot X_+(0),\label{GageCalcT},
\end{align}
where in contrast to Section \ref{Rn}, the boundary terms do not vanish. Withous loss of generality we choose coordinates so that $X_+'(0)=X_+'(L)=(1,0)$ (this is possible by periodicity and Rolle's theorem for $C^1$ functions) and such that
$X_+(L^+)-X_-(L^-)=(1,0)$. Then the above becomes
\begin{align}
 \int p \kappa dS = L^+ - 1.
\end{align}

\noindent Adding and subtracting $\bar \kappa$ to $\kappa$ in the integrand on the left side of \eqref{GageCalcT} as before, we have
\begin{equation}\label{75}
 L^+ - 1 - \bar \kappa \int_0^{L^+} p dS = \int_0^{L^+} p (\kappa - \bar \kappa) dS.
\end{equation}
On each curve $X_{\pm}$ we have
\[ \int_0^{L^{\pm}} \kappa = \alpha(L^{\pm}) - \alpha(0),\]
where $\alpha$ denotes the angle of the normal vector with respect to some reference axis. By periodicity $\alpha(0)=\alpha(L^{\pm})$ and so $\bar \kappa = 0$. We are left with
\begin{equation}\label{76}
 L^+-1 = \int_0^{L^+} (p-\bar p) (\kappa - \bar \kappa).
\end{equation}

We can compute via the Gauss-Green theorem that  $\bar p \geq [2A - L_{\textrm{out}}]/L = [1 - L_{\textrm{out}}]/L$ since $A=1/2$, where the term $-L_{\textrm{out}}/L$ comes from subtracting
the contribution from the boundary of the torus. Applying Cauchy-Schwarz to \eqref{75} we are left with 
\begin{equation}
 \int_{\partial \Omega} p^2 - (\bar p)^2L  \leq \int_{\partial \Omega} p^2 - \frac{(1-L_{\textrm{out}})^2}{L}.
\end{equation} 

\noindent Using Lemma \ref{starlem}  the above is controlled by
\begin{align}\label{T2.2}
 C(L_{\textrm{out}}-L_{\textrm{in}}) + C(L-2) + \int_{\partial S_{L_{\textrm{out}}}} (p^*)^2 dS^* -  \frac{(1-L_{\textrm{out}})^2}{L}.
\end{align} Observe that the difference between the centers of $S_{L_{\textrm{in}}}$ and $S_{L_{\textrm{out}}}$ is controlled
by $|L_{\textrm{out}} - L_{\textrm{in}}|$. Since the function $p^*$ is defined with respect to the center of $S_{L_{\textrm{in}}}$,
it is easily seen via direct computation, setting $X = X + c - c$ so that $O+c$ is the center of $S_{L_{\textrm{out}}}$ that
\[ \int_{\partial S_{L_{\textrm{out}}}} (p^*)^2 dS^* \leq \frac{L_{\textrm{out}}^2}{2} + C(L_{\textrm{out}}-L_{\textrm{in}}),\]
where $C>0$ is universal. Inserting this into \eqref{T2.2} we have
\begin{align} 
\int_{\partial S_{L_{\textrm{out}}}} &(p^*)^2 dS^* -  \frac{(1-L_{\textrm{out}})^2}{L} \leq \frac{2L_{\textrm{out}}^2L-1}{4L}\\ &+ \frac{1}{L} \(\frac{1}{2}-L_{\textrm{out}}\) \leq 2\frac{L-2 + L_{\textrm{out}}-L_{\textrm{in}}}{L} + C(L_{\textrm{out}}-L_{\textrm{in}}).\nonumber
\end{align}
Thus it follows that
\begin{align}
 \int_{\partial \Omega}(p - \bar p)^2 dS \leq C(L-2) + C(L_{\textrm{out}}-L_{\textrm{in}}),
\end{align} 
which when inserted into \eqref{76} and after applying Cauchy-Schwarz yields

\begin{align}
 L^+-1 \leq C \sqrt{ (L-2) + (L_{\textrm{out}}-L_{\textrm{in}})} \sqrt{\int_{\partial \Omega} (\kappa - \bar \kappa)^2 dS}.
\end{align}
Moreover $L-2 = (L^+-1)+(L^--1) \geq |L_{\textrm{out}}-L_{\textrm{in}}|$ (see Figure 2). Repeating the above analysis on the curve $X_-$ and
adding the results, there exists a universal constant $C>0$ such that 
\begin{align}
 L-2 \leq C \sqrt{L-2}\sqrt{\int_{\partial \Omega} (\kappa - \bar \kappa)^2 dS}.
\end{align}
Dividing by $\sqrt{L-2}$ and using $L-2 \geq |L_{\textrm{out}}-L_{\textrm{in}}|$ again we have
\begin{equation}
 |L_{\textrm{out}}-L_{\textrm{in}}|^{1/2} \leq C\sqrt{\int_{\partial \Omega} (\kappa - \bar \kappa)^2 dS},
\end{equation}
which yields the result upon squaring both sides.

\end{proof}

\noindent Now we have the proposition analogous to Proposition \ref{equicontrol}.

\begin{proposition}\label{equicontrolT} Let $\Omega \subset \mathbb{T}^2$ be homeomorphic to $S_{n=1}$. Then there exists a universal constant
$C>0$ such that
 \begin{align}
 \|\phi_{\Omega} - \bar \phi_{\Omega}\|_{L^{\infty}(\partial \Omega)} \leq C |L_{\textrm{out}}-L_{\textrm{in}}|
\end{align}
\end{proposition}

\noindent {\sc Proof of Proposition \ref{nonexist}}

\noindent Consider any two points $y, z \in \partial \Omega$ and assume first $\phi_{\Omega}(y) > \phi_{\Omega}(z)$. Then we have
\begin{align}
 \phi_{\Omega}(y) - \phi_{\Omega}(z) &= \int_{\Omega} G(x,y)+C - G(x,z)-C dx \\
&\leq \int_{S_{L_{\textrm{out}}}}G(x,y)+C  - \int_{S_{L_{\textrm{in}}}} G(x,z)+C,\label{radialeqn2}
\end{align}
where $C$ is chosen so that $G + C > 0$ on $\mathbb{T}^2$. This is possible since $G$ takes the form
\[ G(x,y) = -\frac{1}{2\pi} \log |x-y| + S(x-y),\]
where $S$ is smooth \cite{gilbarg}. 
 The integrals of the logarithmic terms appearing in \eqref{radialeqn2} represent, up to translations, solutions of the one dimensional Poisson equation
\begin{equation}\label{R.633}
-u_{xx} =\left\{\begin{array}{ccc}
1&\mbox{for}&0\le x\le w\\
0&\mbox{for}&w\le x\le 1\\
\end{array}\right.
\end{equation}
for $w = L_{\textrm{in}}$ and $w=L_{\textrm{out}}$. One can explicitly solve \eqref{R.633},
\begin{equation}\label{R.6334}
u(x) =\left\{\begin{array}{ccc}
-\frac{1}{2} x^2 + (w- \frac{1}{2}w^2)x &\mbox{for}&0\le x\le w\\
-\frac{1}{2} w^2 (x-1)&\mbox{for}&w\le x\le 1\\
\end{array}\right. + C_0
\end{equation}
where $C_0$ is a constant. It is then easily seen via direct computation
that there is a universal constant $C>0$ such that the integrals of the logarithmic terms in \eqref{radialeqn2} are controlled by 
\begin{equation}\label{lapsol}
 C(L_{\textrm{out}} - L_{\textrm{in}}).
\end{equation}
 Since the remaining terms in \eqref{radialeqn2} are bounded in $L^{\infty}(\mathbb{T}^2)$ we conclude that 
\begin{align}
\phi_{\Omega}(y)-\phi_{\Omega}(z) \leq  C(L_{\textrm{out}}-L_{\textrm{in}}),
\end{align}
where $C>0$ is universal. Arguing identically when $\phi_{\Omega}(y) \leq \phi_{\Omega}(z)$ we conclude
\[ |\phi_{\Omega}(y) - \phi_{\Omega}(z)| \leq C(L_{\textrm{out}}-L_{\textrm{in}}).\]
Choosing $z$ so that $|\phi_{\Omega}(y)- \bar \phi_{\Omega}| \leq |\phi_{\Omega}(y)- \phi_{\Omega}(z)|$, and
optimizing over $y \in \partial \Omega$ yields the result. $\Box$\\
\medskip

 \noindent {\sc Proof of Theorem \ref{nonexist} iii)}
Now the proof of Theorem \ref{nonexist} follows immediately from combining Proposition \ref{isopermineqT} and Proposition \ref{equicontrolT} as
was done in the case of $U=\mathbb{R}^2$. When $\Omega$ is homeomorphic to the disk, one can repeat the analysis for $U=\mathbb{R}^2$. In particular
line \eqref{lastest1} in the proof of Theorem \ref{nonexist} items i) and ii) will be replaced with
\[ |\phi_{\Omega}(y) - \phi_{\Omega}(z)| \leq C_0 R_{\textrm{out}} (R_{\textrm{out}} - R_{\textrm{in}}),\]
where the constant $C_0$ depends only on the Green's potential $G$ for the torus $\mathbb{T}^2$. In this case we have a uniform bound
on $R_{\textrm{out}}$, and so the result comes from combining the above estimate with Proposition \ref{isopermineq} and using the fact
that $L \geq 2$. 
 Otherwise, if $\Omega$ is homeomorphic to $S_{n=1}$, the constant comes from multiplying the constants in Proposition \ref{isopermineqT} and \ref{equicontrolT}. 

\medskip

In the general case when $X$ is not assumed to be $C^2$ but $\kappa \in L^2(\partial \Omega)$, the result follows by setting $X_{\pm} = X_{\pm}(0) + L^{\pm} \int_0^t e^{i\theta(r)}dr$, mollifying
$\theta^{\epsilon} := \eta^{\epsilon} * \theta$ and passing to the limit in the inequality, using in particular the fact that $(\theta^{\epsilon})' \to \theta'$ in $L^2([0,1])$ and $X^{\epsilon} \to X$ uniformly, and
the fact that the mollified curve will remain star shaped since we mollify $\theta$. We
omit the details.
\section{Proof of Theorems}\label{finalproof}

\medskip

\noindent {\sc Proof of Theorem \ref{main4}:}
Let $C(L,\alpha)$ denote the constant appearing in Theorem \ref{nonexist}. We prove only the first item when $U=\mathbb{R}^2$ since the second item for $U=\mathbb{T}^2$ is argued identically. Observe that variations $\Omega \to \Omega_t$ induced by the normal velocity $\kappa - \bar \kappa$ are admissible
in the sense of Definition \ref{defcp}. Then for any convex set $ \Omega$ with $\kappa \in L^2(\partial \Omega)$ we have
\begin{align}
\frac{dE(\Omega_t)}{dt}\big|_{t=0} &= - \int_{\partial  \Omega} (\kappa - \bar \kappa)^2dS - \int_{\partial \Omega} (\kappa- \bar \kappa )(\phi_{\Omega} - \bar \phi_{\Omega})dS\\
&\leq - \int_{\partial \Omega} (\kappa - \bar \kappa)^2dS + \sqrt{C(L,\alpha)}L^{1/2}\int_{\partial \Omega}(\kappa - \bar \kappa)^2dS  \label{ineq3}\\
&\leq - (1- \bar \eta  \tilde C) \int_{\partial \Omega} (\kappa - \bar \kappa)^2dS,\label{in10}
\end{align}
whenever $\bar \eta < \bar \eta _{cr}$ with $\bar \eta _{cr}$ defined as in Theorem \ref{main4} and where $\tilde C>0$ is universal . The first line is simply the computation of the derivative of $E$ along
the flow \eqref{VPMCF} at $t=0$. Line \eqref{ineq3} follows from Theorem \ref{nonexist} along with Cauchy Schwarz. This establishes the first part of Theorem \ref{main4}, observing that \eqref{in10} implies the lower bound on $\bar \eta _{cr}$. When $\Omega$ is smooth, there exists $t \in [0,T)$ such that the above holds for all $t \in [0,T)$ by Theorem \ref{Huiskenexist2} with $T=+\infty$ when $\Omega$ is convex (cf. Theorem \ref{Huiskenexist}), establishing the second part of Theorem \ref{main4}.  \\
%This establishes immediately Theorem \ref{main3} as we've shown
%that no such set $\Omega$ can be energy critical unless it is a CMC surface. By the result of Alexandrov \cite{Alexandrov}, the only convex, compact CMC surface is the ball. In addition, integrating the above differential inequality and rescaling we conclude
%\begin{equation}
% E(\Omega) \geq E(B) + Cm\int_0^{\infty} \int_{\partial  \Omega_t}  (\kappa - \bar \kappa)^2dt dS ,
%\end{equation}
%for all convex $\Omega$ with $m < m_{cr}$. Using \eqref{areadecrease} along with \eqref{in2} establishes Theorem \ref{main2} $\Box$
\medskip

\noindent {\sc Proof of Theorem \ref{main3}}
We have in fact established Theorem \ref{main3} from the above calculations. Indeed we have shown that
\[ \frac{dE(\Omega_t)}{dt}\big|_{t=0} < 0,\]
where the map $\Omega \mapsto \Omega_t$ is an admissible variation in Definition \ref{defcp}. Therefore by definition $\Omega$ is not a critical point if it does not have constant curvature.
 Using the characterization
of compact, connected, constant curvature curves in $\mathbb{R}^2$ and simply connected constant curvature curves in $\mathbb{T}^2$ establishes the claim. $\Box$\\

\medskip
\noindent {\sc Proof of Theorem \ref{main5}}
We have shown in Theorem \ref{main3} that the only simply connected critical points on $\mathbb{T}^2$ when $\bar \gamma < \bar \gamma_{cr}$ are the ball and the stripe pattern $S_{n=1}$. By the result of Sternberg and Topaloglu \cite{sternberg} $E(\Omega) \geq E(S_{n=1})$ when $\bar \gamma < \bar \gamma_{cr}$ by possibly lowering $\bar \gamma_{cr}$. Moreover since any minimizer has a reduced boundary with $C^{3,\alpha}$ regularity \cite{sternberg}, we have from \eqref{Edecay2} that there is a $T>0$ so that
\begin{equation}
 \frac{dE(\Omega_t)}{dt} \leq -C \int_{\partial \Omega_t} (\kappa - \bar \kappa)^2 dS,
\end{equation}
for all $t \in [0,T)$ and where $C=C(\bar \gamma)>0$. Integrating and using $E(\Omega) \geq E(S_{n=1})$ when $\bar \gamma < \bar \gamma_{cr}$ we have 
\begin{equation}
 E(\Omega) \geq E(S_{n=1}) + C\int_0^{T} \int_{\partial \Omega_t} (\kappa - \bar \kappa)^2 dS =: E(S_{n=1}) + F(\Omega)
\end{equation}
where $F(\Omega)$ is defined implicitly from the above. The result holds for the non-rescaled quantity $\gamma$ since we have the a-priori
bound on $L$ coming from $L \leq \min_{\mathcal{A}_{1/2}} E$. $\Box$\\

\noindent { \sc Counter Example \ref{CE1}}\\

We consider annuli of radii $r, R > 0$ with $r < R$ and let $\Omega = \{ x : r \leq |x| \leq R\}$. Then we
can explicitly calculate $\phi_{\Omega}$, $\bar \phi_{\Omega}$, $\kappa$ and $\bar \kappa$ on $\partial \Omega$ for both
the kernels $K$ and $-\frac{1}{2\pi} \log |x|$. We consider first the logarithmic case. \\

\noindent { \emph{Case 1: $G(x,y) = -\frac{1}{2\pi} \log |x-y|$}}\\

\noindent Solving Poisson's equation $-\Delta \phi_{\Omega} = 1_{\Omega}$ explicitly in radial coordinates we obtain
\begin{align}
\phi_{\Omega}(x) &= \frac{1}{2} (R^2 - |x|^2) - \frac{1}{2} (|x|^2 - r^2) - \frac{R^2}{2} \log R + \frac{r^2}{2} \log r \textrm{ for } x \in [r,R].
\end{align}
It is easily seen that $\bar \kappa = 0$ with $\kappa = \frac{-1}{r}$ on $\partial B_r$ and $\kappa = \frac{1}{R}$ on $\partial B_R$. We first demonstrate
Counter Example \ref{CE1}. We wish to find $R > r > 0$ such that 
\[ \kappa(x) +  \phi_{\Omega}(x) = \kappa(y) +  \phi_{\Omega}(y),\]
for $x \in \partial B_r$, $y \in \partial B_R$. We claim this is the case. Indeed setting
$R=2r$ with $r=\(\frac{1}{2}\)^{1/3}$ we have $\bar \eta  := \(\frac{1}{2}\)^{1/3} (3\pi)^{1/2}(1+ \log (6\pi 2^{-1/3}))$.
Then with these choices of $r$ and $R$,
\begin{equation}
-\frac{1}{r} + \frac{1}{2} (R^2 -r^2) = \frac{1}{R} + \frac{1}{2}(r^2 - R^2).
\end{equation}
holds. Thus $\Omega$ is a solution to \eqref{ELeqn} for 
these choices of $r$ and $R$.\\
\medskip
%We now demonstrate Counter Example \ref{CE2}. After some computations one has
%\begin{align}
%\int_{\partial \Omega} (\kappa - \bar \kappa)^2 dS &= 2\pi \frac{R+r}{rR}\\
%\int_{\partial \Omega} (\phi_{\Omega} - \bar \phi_{\Omega})(\kappa - \bar \kappa) dS &= -2\pi (R^2 - r^2)
%\end{align}
%Let $\Omega_t$ denote the evolution of $\Omega$ under area-preserving curve shortening flow \eqref{VPMCF}. Then we have
%\begin{align}
%\frac{dE(\chi_{\Omega_t})}{dt}\big|_{t=0} = - \int_{\partial \Omega} (\kappa - \bar \kappa)^2 dS -  \int_{\partial \Omega} (\phi_{\Omega} - \bar \phi_{\Omega}) (\kappa - \bar \kappa) dS &= 2\pi (R+r)\( \gamma (R-r) - \frac{1}{rR}\)
%\end{align}
%Set $R= 2r$ and $r=\frac{1}{6\pi}$ as before. The last term in the above equation becomes 
%\begin{align}
% 2\pi (R+r)\( (R-r) - \frac{1}{rR}\) &= \frac{1}{\sqrt{3\pi} (6\pi)^2 rR}  \( m^{1/2}L^2 - \sqrt{3\pi}(6\pi)^2\)\\ &= \frac{1}{\sqrt{3\pi} (6\pi)^2 rR} \( \bar \eta  - \(\sqrt{3\pi}(6\pi)^2 + \sqrt{\pi}2^{-1/6} \log(6\pi 2^{-1/3})\)\).
%\end{align}
% Then when $\bar \eta  \geq \frac{(6\pi)^3}{2}+ \sqrt{\pi}2^{-1/6} \log(6\pi 2^{-1/3}))$, Theorem \ref{main4} fails to hold.  \\

\noindent { \emph{Case 2: $G(x,y) = \frac{1}{|x-y|^{\alpha}}$}}\\

\noindent In this case we have 
\begin{align}
 \phi_{\Omega}(x) &=  R^{2-\alpha} \int_{B(0,1)} \frac{d\bar y}{|x/R- \bar y|^{\alpha}} - r^{2-\alpha}\int_{B(0,1)} \frac{d\bar y}{|x/r- \bar y|^{\alpha}}\\
&= \(\int_{B(0,1)} \frac{dy}{|\tilde x- y|^{\alpha}}\) (R^{2-\alpha} - r^{2-\alpha}) \textrm{ for } \tilde x \in \partial B(0,1).
\end{align}

As in the previous case, it is seen via direct computation that $\Omega$ is a solution to \eqref{ELeqn} when $R=2r$ and 
$\bar \eta  = m^{1/2} L^{2-\alpha} = C_0$ where $C_0$ is an explicit constant, thus establishing Counter Example \ref{CE1} for $G=K$. 

%Similarly to the previous case, we conclude
%\begin{align}
%\frac{dE(\chi_{\Omega_t})}{dt}\big|_{t=0} = - \int_{\partial \Omega} (\kappa - \bar \kappa)^2 dS - \int_{\partial \Omega} (\phi_{\Omega} - \bar \phi_{\Omega}) (\kappa - \bar \kappa) dS &= C\( c \bar \eta - 1\) \geq 0,
%\end{align}
%for constants $c=c(\alpha)$, $C=C(\alpha) >0$, whenever $\bar \eta  c \geq 1$. This establishes Counter Example \ref{CE2} for $G=K$. \\

\medskip

\textbf{Acknowledgements} The author would like to thank his advisor Professor Sylvia Serfaty for her guidance and support during this work and for a careful reading of the preliminary draft. In addition the author would like to thank Professor Gerhard Huisken, Professor Robert Kohn, Christian Seis, Alexander Volkmann and Samu Alanko
 for helpful suggestions and comments and in particular Simon Masnou who explained the generalization of Gauss-Bonnet to non smooth curves. 
The author also extends his gratitude to the Max Planck Institute Golm for an invitation, which was the source of many helpful discussions.  

\end{document}